\documentclass[conference]{ieeeconf}
\IEEEoverridecommandlockouts
 
\usepackage{amsmath,amssymb,amsfonts}
\usepackage{tabularx}
\usepackage{multirow}
\usepackage{longtable}
\usepackage{algorithmicx}
\usepackage{algorithm}
\usepackage{algpseudocode}

\usepackage{flushend}
\usepackage{graphicx}
\usepackage{textcomp}
\usepackage{xcolor}
\usepackage{subfigure}
\usepackage{stfloats}
\usepackage[para,online,flushleft]{threeparttable}
\usepackage{framed}
\usepackage{makecell}
\usepackage{booktabs}
\usepackage{multirow}
\usepackage{url}
\usepackage{cite}
\usepackage{adjustbox}
\usepackage{caption}
\newtheorem{remark}{Remark}
\newtheorem{definition}{\bf Definition}

\newtheorem{problem}{\bf Problem}

\newtheorem{proposition}{\bf Proposition}

\def \F{\mathbf{F}}
\def \G{\mathbf{G}}
\def \U{\mathbf{U}}
\def \s{\mathbf{s}}

\def \hrz{\text{hrz}}

\begin{document}

\title{\bf Synthesis of Temporally-Robust Policies for  \\Signal Temporal Logic Tasks using Reinforcement Learning\\

\thanks{This work was supported by the National Natural Science Foundation of China (62173226,62061136004).}
}

\author{Siqi Wang, Shaoyuan Li, Li Yin, Xiang Yin 
\thanks{Siqi Wang, Shaoyuan Li, and Xiang Yin are with Department of Automation and Key Laboratory of System Control and Information Processing, Shanghai Jiao Tong University, Shanghai 200240, China. 
Li Yin is with the Institute
of Systems Engineering, Macau University of Science and Technology, Taipa, Macao SAR, China.
 E-mail: \tt\small $\{$sq\_wang, syli, yinxiang$\}$@sjtu.edu.cn}
}

\maketitle
\thispagestyle{empty}
\pagestyle{empty}
\setlength{\abovecaptionskip}{0pt}
\setlength{\belowcaptionskip}{0pt}
\setlength{\textfloatsep}{6pt}
\begin{abstract}
This paper investigates the problem of designing control policies that satisfy high-level specifications described by signal temporal logic (STL) in unknown, stochastic environments. 
While many existing works concentrate on optimizing the spatial robustness of a system, our work takes a step further by also considering \emph{temporal robustness} as a critical metric to quantify the tolerance of time uncertainty in STL. 
To this end, we formulate two relevant control objectives to enhance the temporal robustness of the synthesized policies. The first objective is to maximize the probability of being temporally robust for a given threshold. The second objective is to maximize the worst-case spatial robustness value within a bounded time shift. We use reinforcement learning to solve both control synthesis problems for unknown systems. Specifically, we approximate both control objectives in a way that enables us to apply the standard Q-learning algorithm. Theoretical bounds in terms of the approximations are also derived. We present case studies to demonstrate the feasibility of our approach.
\end{abstract}

\vspace{-3pt}
\section{Introduction} 
\vspace{-2pt}
Autonomous systems operating in dynamic environments face the challenge of making complex real-time decisions. These systems, known as \emph{time-critical systems}, must process real-time information to achieve their goals, and the accuracy of their decisions is crucial, especially with temporal constraints involved. For example, an automated guided vehicle may need to retrieve a workpiece within $10$ minutes and return it within $20$ minutes. However, an ad-hoc approach for real-time decision-making may lead to errors. Consequently, ensuring formal guarantees for real-time systems has gained focus in recent years.

Signal temporal logic (STL) is a formal specification language used to describe high-level temporal behaviors of continuous signals. It extends metric temporal logic for real-time systems by incorporating real-valued predicates on signals \cite{maler2004monitoring,donze2013efficient,yu2024model}. One major advantage of STL is its ability to provide both Boolean satisfaction and quantitative measures, termed spatial robustness degree. This unique feature has led to the spreading use of STL in cyber-physical systems, including autonomous robots \cite{silano2021power}, process control systems \cite{farahani2018formal}, smart cities \cite{ma2021novel} and self-driving vehicles \cite{hekmatnejad2019encoding}.

The spatial robustness essentially quantifies the  satisfaction of STL tasks based on value changes in the predicate function. However, practical scenarios also involve signal delays or ahead-of-time occurrences during online executions, which necessitates exploring STL satisfaction under signal time uncertainty. To tackle this, \emph{temporal robustness} was introduced in the literature, which quantifies the maximum left or right time shift a signal trajectory can endure to maintain satisfaction or violation of an STL specification \cite{donze2010robust}. In \cite{rodionova2022combined}, the authors merged left and right temporal robustness and solved control synthesis. 
In \cite{rodionova2022temporal, lindemann2022temporal,yu2023efficient}, the temporal robustness was further investigated for synchronization issues of multi-dimensional signals. Additionally, there are other measures for quantifying the robust satisfaction of STL formulae under time uncertainty, such as conformance \cite{deshmukh2017quantifying} or AverageSTL robustness \cite{akazaki2015time,lin2020optimization}.

% \SQ{is MILP/CBFs used correctly?}
In open or reactive environments, synthesizing control sequences or policies to ensure the satisfaction of STL tasks is a major challenge. To tackle this, various synthesis methods have emerged, including encoding the STL satisfaction as constraints in a mixed-integer linear program (MILP) \cite{raman2014model,kurtz2022mixed,sun2022multi,yu2023model} and capturing the satisfaction regions of the STL formula using control barrier functions (CBFs)\cite{lindemann2018control,lindemann2020barrier,xiao2021high}. 
These approaches assume system knowledge; yet in many applications, the system's dynamic is unknown, and trajectories can only be generated through interactions. 
Thus, reinforcement learning techniques are employed for STL task control synthesis; see, e.g., \cite{aksaray2016q,balakrishnan2019structured,kalagarla2021model,venkataraman2020tractable,ikemoto2022deep}. In \cite{aksaray2016q}, Q-learning was applied to maximize the expected robust degree for an unknown system modeled by Markov decision processes (MDPs). Moreover, \cite{venkataraman2020tractable} offered more efficient MDP construction methods.

The focus of the aforementioned work is to enhance spatial robustness of control policies. However, in the context of temporal robustness for time-uncertain systems, synthesis methods are only available in recent works such as \cite{rodionova2021time, buyukkocak2022temporal, rodionova2022combined, rodionova2022temporal}. These approaches have a common feature, which is to extend the MILP-based approach by encoding the temporal robustness using new variables. Similar to the issue of spatial robustness, these approaches rely on the system model. To the best of our knowledge, optimizing STL task temporal robustness in systems with unknown dynamics remains unexplored.

In this paper, we address the challenge of control policy synthesis for unknown stochastic systems to achieve STL tasks. Unlike previous work focusing on spatial robustness, we extend our focus to include temporal robustness. Specifically, we tackle two control synthesis problems. Firstly, we aim to synthesize a policy that maximizes the probability of the trajectory's temporal robustness exceeding a set threshold. Secondly, we aim to maximize the expected worst-case spatial robustness value of the system's trajectories under time uncertainty. To tackle both problems, we apply reinforcement learning techniques, specifically Q-learning method, for MDPs. Our approach draws inspiration from \cite{aksaray2016q}, who uses a $\tau$-MDP structure for reinforcement learning for spatial robustness. Here we reformulate the two problems and select appropriate augmented horizons to handle temporal robustness metrics. We also establish a formal connection between the original problems and the reformulated problems. Our experimental results show that the synthesized policies can effectively enhance the temporal robustness of the system.

\section{Preliminaries}\label{sec:Pre}

This section reviews some basic concepts of signal temporal logic and reinforcement learning of unknown MDPs.

\subsection{Signal Temporal Logic Basics}

Signal temporal logic (STL) is a formal language used for specifying  temporal properties  for real-time systems. 
It is evaluated over dense-time signals in continuous metric space $\mathbb{R}^m$. The syntax of STL is recursively defined as follows \cite{maler2004monitoring}:
\begin{equation}\label{def:STL-general}
    \phi ::= \textsf{true} \mid \mu\mid \neg\phi \mid \phi_1 \land \phi_2 \mid \phi_1\U_{[a,b)}\phi_2, 
\end{equation}
where $\mu:\mathbb{R}^m\to \{\textsf{true},\textsf{false}\}$ is an  atomic predicate
such that it is satisfied when the value of the associated predicate function $h^\mu(s)>0$, where $s\in\mathbb{R}^m$; 
$\neg$ and $\wedge$ are the standard Boolean operators, negation and conjunction respectively, and   
$\U_{[a,b)}$  is the temporal operator ``until" with $a<b$ and $a,b \in \mathbb{N}$. Furthermore, one can  induce temporal operators:
\begin{itemize}
    \item 
    ``eventually" by $\F_{[a,b)}\phi := \textsf{true} \U_{[a,b)}\phi$; and  
    \item 
    ``always" by $\G_{[a,b)}\phi := \neg \F_{[a,b)}\neg \phi$.
\end{itemize}

\begin{definition}[Spatial Robustness of STL] 
Let $\phi$ be an STL formula,  $\s=s_0s_1\cdots$ be a signal and $t\in \mathbb{N}$ be a time instant. The \emph{spatial robustness} of $\phi$ w.r.t.\  $\s$ at time $t$, denoted by $\rho(\phi,\s,t)$, is defined recursively by
\begin{equation}
    \begin{array}{lll}
        \rho(\textsf{true},\s,t) &=& +\infty, \\
      \rho(\mu,\s,t) &=& h^\mu(s_t), \\
        \rho(\neg \phi,\s, t) &=& -\rho(\phi,\s,t),\\
        \rho(\phi_1\land\phi_2,\s,t)&=& \min\left(\rho(\phi_1,s,t),\rho(\phi_2,s,t)\right),\\ 
        \rho(\phi_1\U_{ [a,b)}\phi_2,\s,t)&=&\max_{t'\in[a+t,b+t)}\min\{\rho(\phi_2,\s,t'),\\ & &\min_{t''\in[t,t')}\rho(\phi_1,\s,t'')\}. 
        \vspace{-8pt}
    \end{array}
\end{equation}
\end{definition}

The Boolean semantic of STL is a special instance of the spatial robustness semantic. 
Let $\phi$ be an STL formula, $\s=s_0s_1\cdots$ be a signal and $t\in \mathbb{N}$ be a time instant. We say $\phi$ is  satisfied by $\s$ at $t$
, denoted by $\s[t]\models \phi$ if its robust value is  larger than zero, i.e., 
\vspace{-3pt}
\begin{equation}
\s[t]\models \phi \Leftrightarrow  \rho(\phi,\s,t)>0.
\label{eqn:boolean}
\end{equation}\vspace{-3pt}
We also define the \emph{characteristic function} of the STL formula $\phi$ w.r.t.\  signal $\s$ at time instant $t$ by  
\begin{align} 
	\mathcal{X}(\phi,\s,t) = 
		\left\{
		\begin{array}{rl}
			1 &  \text{if}\quad \s[t] \models \phi  \\
			-1&  \text{otherwise}
		\end{array}
		\right.   .
\end{align} 

For any STL formula, its satisfaction as well as the robust degree can be completely determined within its \emph{horizon} denoted by  $\hrz(\phi)$, which can be computed as the maximum sum of the time interval bound of all nested temporal operators; see, e.g.,\cite{aksaray2016q}.
 
Hereafter in this work, we will restrict our attention to the following fragment of STL formulae:
\begin{equation}\label{def:STL-frag}
    \Phi::=\F_{[0,H)}\phi \mid \G_{[0,H)}\phi,  
\end{equation}
where $\phi$ is a general STL formula  in Equation~\eqref{def:STL-general} and $H\in \mathbb{N}$ is a time instant.
\vspace{-3pt}
\subsection{Temporal Robustness of STL}
\vspace{-2pt}
The spatial robustness semantic quantifies the extent of STL satisfaction based on  predicate function value changes. 
However, it cannot capture the robust satisfaction of a formula concerning time shifts. 
To address this, \emph{temporal robustness} was introduced in \cite{donze2010robust} to quantify the maximum left or right time shift a signal trajectory can bear to maintain satisfaction or violation of an STL specification. 

\begin{definition}[Temporal Robustness of STL, \cite{donze2010robust}]
Let $\phi$ be an STL formula,  $\s=s_0s_1\cdots$ be a signal and $t\in \mathbb{N}$ be a time instant.
The left (respectively, right) temporal robustness, denoted by $\theta^-(\phi,\s,t)$ (respectively, $\theta^+(\phi,\s,t)$), is defined as the maximum left (respectively, right) time shift for a signal $\s$ to maintain the satisfaction or the violation of  STL formula $\phi$. Formally, we have
\begin{equation}
\small
\begin{array}{ll}
    \theta^+(\phi,\s,t)&=
    \mathcal{X}(\phi,\s,t)\cdot
    \max\left\{
    d \,\middle\vert\,\!\!\!\!\!\!\!\!\!
    \begin{array}{cc}
         &   \forall t'\in [t,t+d] \\
         &  \mathcal{X}(\phi,\s,t')=\mathcal{X}(\phi,\s,t)
    \end{array}\!\!
    \right\}\\
     \theta^-(\phi,\s,t)&=
    \mathcal{X}(\phi,\s,t)\cdot
    \max\left\{
    d \,\middle\vert\,\!\!\!\!\!\!\!\!\!
    \begin{array}{cc}
         &   \forall t'\in [t-d,t] \\
         &  \mathcal{X}(\phi,\s,t')=\mathcal{X}(\phi,\s,t)
    \end{array}
    \right\}
\end{array}    
\end{equation}

\end{definition}

\begin{remark}

In the provided definition, when signal shifts result in undefined states, we fill these states with the initial or final values.
For example, for all $t<0$, we have $s_t:=s_0$.    
\end{remark}

\begin{remark}
For simplicity, we only consider left robustness here, given that right shifts or time-delays are more common in real-world scenarios. Hereafter, we will simply use the terminology "temporal robustness" denoted as $\theta(\cdot)$, to stand for left temporal robustness $\theta^-(\cdot)$.
\end{remark}

\vspace{-3pt}
\subsection{Reinforcement Learning for MDPs}\label{sec:RL}
We model the underlying dynamic system by a Markov decision process. 
Formally, an MDP is a  tuple $M=\left(\Sigma, s_0,A, P, R\right)$, 
where $\Sigma$ is the state space, $s_0$ is the initial state, and $A$ is the action space,  $P: \Sigma\times A\times\Sigma\to [0,1]$ is a transition probability function and $R:\Sigma\to \mathbb{R}$ is the reward function. We assume that the state space and action space are known, but the transition probability is \emph{unknown}.  
For simplicity, we assume the initial state is unique; however, this can easily be extended to the scenario with an initial distribution.
  
Given MDP $M$, a (stationary) control policy is a function $\pi: \Sigma \times A\to [0,1]$ which assigns probabilities to actions at each state such that for all $s\in \Sigma$, $\sum_{a\in A}\pi(s,a)=1$. The objective of reinforcement learning is to synthesize a control policy that  maximizes the total sum of discounted rewards through simulation data \cite{sutton2018reinforcement}, i.e., 
\vspace{-5pt}\begin{equation}
 \pi^* = \arg\max\limits_{\pi}\mathbb{E}\sum_{t=0}^{T}\gamma^t R_t,    
\end{equation}
where $R_t$ is the random variable for the reward at instant $t$ when the agent applies policy $\pi$ and $\gamma\in[0,1]$ is a discount factor that balances future rewards.

\emph{Q-learning} is a prominent algorithm in reinforcement learning 
to achieve the above objective without knowing the transition probability. It is a model-free, off-policy and temporal-difference method utilizing a Q-table to hold values for state-action pairs.
At each instant $t$, with $\alpha$ being the learning rate, we use the one-step transition data $(s_t,a_t,r_t,s_{t+1})$  to update the Q-table according to the Bellman Equation as follows:
\begin{equation}\small
     Q(s_t,a_t) := (1-\alpha)Q(s_t,a_t)
    + \alpha[r_t+\gamma \max\limits_{a\in A} Q(s_{t+1},a)].
\end{equation}

\section{Problem Formulation}\label{sec:problem}
We define two optimality metrics for the synthesized control policy by taking the temporal robustness value into account and formulate the control synthesis problems. 

\subsection{Case of Guaranteed Temporal Robustness}
In the first approach, we consider $\delta>0$ as the minimum required value for the temporal robustness of the signal. 
Only signals with temporal robustness greater than or equal to $\delta$
are considered ``robust-enough" signals. 
The following first problem aims to maximize the probability of robust-enough signals. 
\begin{problem}\textbf{(Maximizing Probability with Temporal Robustness Guarantees):} \label{problem:1}
Let
$M$ be an MDP with unknown $P$ 
and $\Phi$ be an STL formula in form of Equation~\eqref{def:STL-frag} with time horizon $\hrz(\Phi)=T+1$. Given $\delta$, find a policy $\pi_1^*$ that maximizes the probability of generating ``robust-enough'' trajectories. That is 
\vspace{-5pt}
\begin{equation}\label{eq:first-prob}
\pi_1^{*}=\arg \max_{\pi}\text{Pr}^\pi[\theta (\Phi,\s_{0:T} )\geq \delta], 
\end{equation} 
where $\Pr^\pi(\cdot)$ denotes the probability under $\pi$.
\end{problem}

Note that, using indicator function $I(\cdot)$,  Equation~\eqref{eq:first-prob} can be expressed equivalently as 
\begin{equation}\label{eq-reform}
     {\pi_1^{*}}=\arg \underset{\pi }{\mathop{\max }}\,{{\mathbb{E}}^{\pi }}[I(\theta (\Phi,\s_{0:T} )\geq \delta )]. 
\end{equation}

\subsection{Case of Spatial-Temporal Robustness}
Note that, the above problem formulation does not consider spatial details within signal values.  To address this, we extend it to a spatial-temporal robustness joint optimization problem. In this formulation, our goal is twofold: maximizing the satisfaction probability under time uncertainty and enhancing the extent of satisfaction.

Specifically, we consider $\delta>0$ as an upper bound of possible time shifts. 
Then we define the worst-case spatial robustness value with time shifts bounded by $\delta$ by
\begin{equation}
\rho_\delta(\Phi,\s_{0:T}):=\min\{\rho(\Phi,\s_{0:T},d)\mid d\leq \delta\}.
\end{equation}
Our objective is to maximize the expectation of such worst-case spatial robustness value. 

\begin{problem}\label{problem:worst}
\textbf{(Maximizing Expectation of Spatio-Temporal Robustness):} Under the same setting as Problem~\ref{problem:1}, find a control policy $\pi_2^*$ such that
\begin{equation}
     \pi_2^*=\arg\max\limits_\pi \mathbb E^\pi\left[\rho_\delta(\Phi,\s_{0:T})\right].
\end{equation}
\end{problem}

% \begin{remark}
% Problem~2 can be considered as a  quantitative extension of Problem~1. 
% Specifically, Problem~2 aims to maximize the expectation of quantitative value $\rho_\delta(\Phi,\s_{0:T})$.  
% However,  according to Equation~\eqref{eq-reform}, Problem~1 only aims to maximize the probability that $\rho_\delta(\Phi,\s_{0:T})$ is positive since we have  $I(\theta(\Phi,\s_{0:T})\geq \delta)=1$ iff $\rho_\delta(\Phi,\s_{0:T})>0$.
% \end{remark} 

\section{Reinforcement Learning for Temporal Robustness}\label{sec:Method}

In this section, we reshape the previous problems into RL-friendly ones. 
Specifically, we approximate the original objective functions in summation-based forms for which standard Q-learning algorithms can be applied. 
Theoretical bounds between the original problems and their approximations are established.

\subsection{Construction of $\tau$-MDPs}
As noted by \cite{aksaray2016q}, for STL spatial robustness, the standard Q-learning algorithm cannot be applied to the original MDP since the reward function is non-Markovian. 
The same issue also exists for temporal robustness, and time shifts further add complexity.  
To address this, we use the $\tau$-MDP structure proposed in \cite{aksaray2016q}, which essentially augments states in the original MDP by a sequence of states of length $\tau$.

\begin{definition}[$\tau$-MDP, \cite{aksaray2016q}] 
Given an MDP $M=\left(\Sigma,s_0,A,P,R\right)$ and a positive integer $\tau\in\mathbb{N}$, its associated  $\tau$-MDP is a tuple
$
M^\tau=\left(\Sigma^\tau,s^\tau_0,A,P^\tau,R^\tau\right)
$, where
    \begin{itemize} 
    \item $\Sigma^\tau \subseteq \left(\Sigma \right)^\tau$ is the set of states.
    \item $s^\tau_0$ is the initial state, which is initialized as a string $s_0s_0\dots s_0$ of length $\tau$;
    \item $A$ is the action space, which is the same as $M$; 
    \item $P^\tau: \Sigma^\tau\times A \times\Sigma^\tau\to [0,1] $ is the transition probability function such that, 
      for any $s^\tau_t,s^\tau_{t+1}\in \Sigma^{\tau},a_t\in A$, 
      if $s^\tau_{t+1}(i) = s^\tau_t(i+1), \forall i\in\{0,1,\cdots,\tau-2\}$, where $s^\tau_{t}(i)$ denotes the $i$th element in $s^\tau_{t}$,
      then we have
      $
      P^\tau(s^\tau_t,a_t,s^\tau_{t+1}) =  P(s^\tau_t(\tau-1),a_t,s^\tau_{t+1}(\tau-1)).
      $
      Otherwise, $P^\tau(s^\tau_t,a_t,s^\tau_{t+1})=0$;
        \item 
        $R^\tau:\Sigma^\tau \rightarrow \mathbb{R}$ is the reward function defined over the $\tau$-MDP.
        % %= R(s^t,a_t)$, where $R(\cdot)$ is the reward function of the primal MDP. 
        % \XY{what is $s^t$? I think the reward should be constructed later}
        % \SQ{$R(s_t,a_t)$ should be deleted, it doesn't make sense to construct reward from primal states.}
    \end{itemize}
\end{definition}

Intuitively, the $\tau$-MDP ``unfolds" the original MDP by retaining  the last $\tau$ visited states. 
That is, each $s^\tau_t$ denotes the $\tau$-step history, i.e., $s^\tau_t = \s_{t-\tau+1:t}$, where $\tau$ is determined by the task.
For a general STL formula, we need to choose $\tau$ as the entire horizon of the formula to gather enough information for robustness computation. 
Yet, for formulae $\F_{[0,H)}\phi $ or $\G_{[0,H)}\phi$, we only need $\tau = \hrz(\phi) + \delta$. 
Specifically, we need  $\hrz(\phi)$-step information to determine the satisfaction of the internal sub-formula $\phi$, and an additional $\delta$-step information to determine the temporal robustness. 

\subsection{Approximation of Robust Probability}

Before addressing Problem \ref{problem:1}, where our objective is to maximize the probability of being temporally robust to a time shift threshold, we first extend the STL semantic to delayed partial signals. 

\begin{definition}
    \textbf{(Satisfaction of Delayed Signal):} Given specification $\Phi$, $\tau$-state $s^\tau_t$, and a small delay amount $d<\delta$, we denote the delayed (right shifted) signal trace as $\s_{t-h-d+1:t-d}$, where $h=\hrz(\phi)$. The satisfaction of delayed signal is denoted as follows: 
    \begin{IEEEeqnarray}{lCr}
        \text{sat}(\phi, s^\tau_t ,d)=\left[\s_{t-h-d+1:t-d} \models \phi'\right],
    \end{IEEEeqnarray} 
    where $\phi'$ is obtained from  $\phi$ by right-shifting the effective time window by $t-h-d+1$ steps, e.g., for $\phi=\G_{[0,h)} (s< C)$, we have $\phi'=\G_{[t-h-d+1,t-d)}(s<C)$.
    In another word, the satisfaction of delayed signal is determined by whether every (or at least) one state on trace $\s_{t-h-d+1:t-d}$ satisfies the inner sub-formula $\phi'$ by Boolean semantic (\ref{eqn:boolean}). With a slight abuse of notation, hereafter, sub-formulae are by default shifted according to the elapsed time $t$ and thus we do not differentiate $\phi$ and $\phi'$.
\end{definition}

Having defined satisfaction of delayed signals, we can now formally qualify the concept of being temporally robust. 

\begin{definition}\label{def:being}
    \textbf{($\delta$-Temporally-Robust):} Given specification $\Phi$, $\tau$-state $s^\tau_t$ and temporal robustness lower bound $\delta$, 
    we  say $\tau$-state $s^\tau_t$ is $\delta$-temporally-robust if  $s^\tau_t$ satisfies the specification $\Phi$ for any delay $d\in[0,\delta]$, and we denote
    \begin{IEEEeqnarray}{lCr}
        \text{rb}(\Phi ,s_{t}^{\tau },\delta)=
        \left\{ \begin{matrix}
            1 & \text{if}\underset{0\le d\le \delta }{\mathop{\min }}\,(\text{sat}(\Phi ,s_{t}^{\tau },d))=1  \\
            0 & \text{otherwise}  \\
         \end{matrix} \right.
    \end{IEEEeqnarray}
\end{definition}

\begin{proposition} The temporal robustness guarantee of signal $\s_{0:T}$ w.r.t.\ the overall formula $\Phi$ can be determined by the temporal robustness guarantee of partial signal $s^\tau_t$ w.r.t the sub-formula $\phi$, i.e., the following equation holds: 
\begin{equation}
\small
\begin{array}{rl}
    &I(\theta(\Phi,\s_{0:T})\geq\delta)=\\
        &\quad\quad\quad\quad \left\{
            \begin{array}{ll}
                \max\limits_{t=0:T}\left(\text{rb}(\phi,s^\tau_t,\delta)\right) & \text{if } \Phi=\F_{[0,H)}\phi\\
                \min\limits_{t=0:T}\left(\text{rb}(\phi,s^\tau_t,\delta)\right) & \text{if } \Phi=\G_{[0,H)}\phi
            \end{array}
            \right.
\end{array}
\end{equation}
\end{proposition}
\vspace{5pt}
\begin{proof}
    For disjunction and conjunction of sub-formulae, we have the following properties in terms of temporal robustness; see \cite{rodionova2021time}:
    \begin{equation}
         \begin{array}{c}
        \theta(\phi_1\lor\phi_2,\s_{0:T})=\max\left(\theta(\phi_1,\s_{0:T}),\theta(\phi_2,\s_{0:T})\right)\\
        \theta(\phi_1\land\phi_2,\s_{0:T})=\min\left(\theta(\phi_1,\s_{0:T}),\theta(\phi_2,\s_{0:T})\right) 
        \end{array}
    \end{equation}
For STL formula $\Phi=\F_{[0,H)}\phi$, we use the operator ``$\lor$" to break down the formula, and since $\theta(\phi,\s_{t:t+h-1})\geq \delta, \forall t\in[0,T]$, we have
    \begin{IEEEeqnarray}{lCr}
\vspace{-5pt}        \begin{array}{l}
           % \Phi=\lor_{t=0:T}\phi \\
            \theta(\Phi,\s_{0:T})=\max\limits_{t=0:T}(\theta(\phi,\s_{t:t+h-1})).
        \end{array}
    \end{IEEEeqnarray}
Since $I(\max\limits_{t=0:T}(\theta(\phi,\s_{t:t+h-1}))\geq \delta)$ is equivalent to $\max\limits_{0:T}(\text{rb}(\phi,s^\tau_t,\delta))$, we obtain that $$I(\theta(\Phi,\s_{0:T})\geq\delta)=\max\limits_{t=0:T}\left(\text{rb}(\phi,s^\tau_t,\delta)\right).$$ 
Proof is similar for $\Phi=\G_{[0,H)}\phi$, as $\min\limits_{t=0:T}\left(\text{rb}(\phi,s^\tau_t,\delta)\right)$ is equivalent to \small $I(\min\limits_{t=0:T}(\theta(\phi,\s_{t:t+h-1}))\geq \delta)$.
\end{proof}

We have established the relation between the temporal robustness of entire trajectory w.r.t.\ the overall specification $\Phi$ and the temporal robustness of partial trajectory w.r.t.\ the inner sub-formula $\phi$.
In \cite{aksaray2016q}, the author approximated and decomposed the objective function using the LSE (log-sum-exp) method into a sum of step rewards. Here we use the similar philosophy in reformulating our problem. The LSE, also known as  a smooth approximation to the maximum function, is defined as follows:
\begin{IEEEeqnarray}{lCr}\label{eqn:LSE}
    \begin{array}{lll}
        \max\left(x_1,\cdots,x_n\right)&\approx &\frac{1}{\beta}\log\sum\limits_{i=1}^{n}e^{\beta x_i}.
    \end{array}
\end{IEEEeqnarray}
The approximation is bounded by the following inequalities:
\begin{IEEEeqnarray}{lCr}
    \begin{array}{ll}
        \max(x_1,\dots,x_n)&\leq \frac{1}{\beta}\log \sum\limits_{i=1}^n e^{\beta x_i}\\
        &\leq \max(x_1,\dots,x_n)\!+\!\frac{1}{\beta}\log n .
    \end{array}
    \label{eqn:bound}
\end{IEEEeqnarray}
The same can be derived for $\min\left(x_1,\cdots,x_n\right)=-\max\left(-x_1,\cdots,-x_n\right)=\approx -\frac{1}{\beta}\log\sum\limits_{i=1}^{n}e^{-\beta x_i}$. Clearly, increasing $\beta$ will narrow the gap between the approximated value and the actual value. 
 
\renewcommand\theproblem{\arabic{problem}A}
\setcounter{problem}{0}
\begin{problem}
    \textbf{(Maximizing Approximated Probability of Being Temporally Robust):} Consider an MDP  $M=\left\langle \Sigma ,s_0,A,P,R \right\rangle $ with unknown $P$, given an STL specification $\Phi$, find a control policy $\pi_{1A}^*$ that maximizes the following objective function:  \label{problem:approx} 
\begin{equation}
    \small{
    \pi^*_{1A}=\left\{\begin{array}{ll}
        \arg\max\limits_{\pi} \mathbb{E}[\sum\limits_{t=0}^{T}e^{\beta\cdot \text{rb}(\phi,s^\tau_t,\delta)}] & \text{if } \Phi=\F_{[0,H)}\phi \\
        \arg\max\limits_{\pi} \mathbb{E}[\sum\limits_{t=0}^{T}-e^{-\beta\cdot \text{rb}(\phi,s^\tau_t,\delta)}] & \text{if } \Phi=\G_{[0,H)}\phi
    \end{array}\right.
    }
\end{equation}

\end{problem}

% \begin{remark}
% By dropping the $\frac{1}{\beta}\log(\cdot)$ component, the LSE decomposition has a shadow lifting effect, i.e., all steps have effects on the total reward gain instead of few key steps that determines the satisfaction of the overall task, intrinsically boosting temporal robustness. This feature has been shown implicitly in the experiment results of \cite{aksaray2016q}. 
% \end{remark}

The following result  shows that the reformulated problem can be arbitrarily close to the original problem. 
\begin{proposition}When $\beta$ goes to infinity, the optimal policy to the approximated problem $\pi^*_{\ref{problem:approx}}$ will converge to the optimal policy to the original problem $\pi^*_{\ref{problem:1}}$. 
\end{proposition}

\begin{proof}
The proposition can be manifested from Equation~\eqref{eqn:bound}, which leads to:
    \begin{equation}\small
        \begin{array}{l}
         \text{Pr}^{\pi_1^*}(\theta(\Phi,\s_{0:T})\geq\delta)\leq
        \text{Pr}^{\pi_{1A}^*}(\theta(\Phi,\s_{0:T}) \geq \delta)\\\leq \text{Pr}^{\pi_1^*}(\theta(\Phi,\s_{0:T})\geq \delta)+\frac{1}{\beta} \log (T+1).
    \end{array}
    \vspace{-10pt}
    \end{equation}
\end{proof}

Note that, when $\beta$ is too large, the objective function will become less smooth, which prolonging the learning convergence\cite{sutton2018reinforcement}. 
Therefore, we select a reasonably large $\beta$ in the experiments.

Now Problem \ref{problem:approx} is in the standard form of reinforcement learning, depending on the type of the STL specification, the step reward is given as:
\begin{IEEEeqnarray}{lCr}\small
    r^\tau_{t} = R^\tau(s^\tau_{t+1}) = \left\{\begin{matrix}
        e^{\beta\cdot\text{rb}(\phi,s_{t+1}^\tau,\delta)} & \text{if } \Phi=\F_{[0,H)}\phi\\
        -e^{-\beta\cdot\text{rb}(\phi,s_{t+1}^\tau,\delta)} & \text{if } \Phi=\G_{[0,H)}\phi
    \end{matrix} \right..
\end{IEEEeqnarray}
\subsection{Maximizing  Spatial-Temporal Robustness}
Although our original purpose is to further consider the spatial robustness in addition to the temporal robustness requirement, the purported approach is also expected to expedite the learning process.  
This is because relying solely on Boolean satisfaction offers little insight into the quality of the current $\tau$-state.
For example, for $\Phi=\F_{[0,H)}\phi$, before reaching a satisfying $\tau$-state, the agent receives minimum reward, resulting in infrequent updates to the corresponding entries in the Q-table. Thus the learning process reduces to a Monte-Carlo search, which is far from efficient. However, incorporating spatial semantic can enhance our data efficiency.
To this end, we denote the objective function in Problem \ref{problem:worst} as $J_0$ and we obtain the following equivalent form regarding $\Phi=\F_{[0,H)}\phi$ and $\Phi=\G_{[0,H)}\phi$  
\begin{equation}\label{eqn:j0}
    \small{
    J_0=\left\{
        \begin{array}{ll}
            \mathbb{E}[ \min\limits_{d\leq \delta} \{\max\limits_{t}\{ \rho(\phi,s^\tau_t,d)\}\}] & 
            \text{if } \Phi=\F_{[0,H)}\phi\\
            \mathbb{E}[ \min\limits_{d\leq \delta} \{\min\limits_{t}\{ \rho(\phi,s^\tau_t,d)\}\}] &
             \text{if } \Phi=\G_{[0,H)}\phi  
        \end{array}
    \right.
    }
\end{equation}

Using the same LSE technique to decompose the \texttt{max} and \texttt{min} operator in $J_0$, we obtain the approximated objective function $J_1$ by 

\begin{equation}\vspace{-5pt}
\small{
    J_1 = \left\{
        \begin{array}{ll}
            \mathbb{E}[\min\limits_{d\leq \delta} \{\sum\limits_{t=0}^T e^{\beta\rho(\phi,s_t^\tau,d)}\}] &\text{if } \Phi=\F_{[0,H)}\phi\\
            \mathbb{E}[\min\limits_{d\leq \delta} \{\sum\limits_{t=0}^T -e^{-\beta\rho(\phi,s_t^\tau,d)}\}] &\text{if } \Phi=\G_{[0,H)}\phi  
        \end{array}
    \right. }
    \vspace{-5pt}
\end{equation}

Note that the above objective function $J_1$ is still not in the additive form, to further obtain the instant reward at each step, we define the objective function $J_2$ by

\begin{equation}\label{eqn:j2}\small{
J_2 = \left\{
        \begin{array}{ll}
            \mathbb{E}[\sum\limits_{t=0}^{T}e^{\beta\cdot \min\limits_{d\leq\delta}\{\rho(\phi,s_{t}^{\tau},d)\}}] &\text{if } \Phi=\F_{[0,H)}\phi\\
            \mathbb{E}[\sum\limits_{t=0}^{T}-e^{-\beta\cdot \min\limits_{d\leq\delta}\{\rho(\phi,s_{t}^{\tau},d)\}}] &\text{if } \Phi=\G_{[0,H)}\phi  
        \end{array}
    \right.}
    \vspace{-5pt}
\end{equation}

We use the objective function $J_2$ to formulate the approximated problem as follows. 
\setcounter{problem}{1}
\begin{problem}
    \textbf{(Maximizing Approximated Expected Spatial-Temporal Robustness):}  \label{problem:approx2}
    Consider an  MDP $M=\left\langle \Sigma ,s_0,A,P,R \right\rangle $ with unknown $P$, given an STL specification $\Phi$, find a control policy $\pi_{2A}^*$ that maximize the objective function $J_2$.
\end{problem}

Specifically, the step reward is given by   
\vspace{-3pt}
\begin{IEEEeqnarray}{lCr}\small
    R^\tau =\left\{ \begin{array}{ll}
        e^{\beta\cdot \min\limits_{d\leq\delta}\{\rho(\phi,s_{t}^{\tau},d)\}} & \text{if } \Phi=\F_{[0,H)}\phi \\
        -e^{-\beta\cdot \min\limits_{d\leq\delta}\{\rho(\phi,s_{t}^{\tau},d\}} & \text{if } \Phi=\G_{[0,H)}\phi
    \end{array}\right. 
\end{IEEEeqnarray}
The following result shows that the optimal value of the approximated Problem~\ref{problem:approx2} provides a lower bound for the optimal value of the original  Problem~\ref{problem:worst}.

\begin{proposition}
    Maximizing the objective function of Problem~\ref{problem:approx2} maximizes the objective function of Problem~\ref{problem:worst}.
\end{proposition}

\begin{proof}
    For $\Phi = \F_{[0,H)}\phi$, we denote $e^{\beta\rho(\phi,s^\tau_t,d)}$ as $M^d_t$, where $d\in\{0,1,\dots,\delta\}$.  
    The approximated objective function $J_1$ can be written as
    \begin{equation}\label{eqn:j1}\small\vspace{-5pt}
        J_1 = \min\left\{\sum_{t=0}^{T}M_t^0,\sum_{t=0}^{T}M_t^1,\dots,\sum_{t=0}^{T}M_t^\delta\right\}.
    \end{equation}
    Also, the objective in Problem \ref{problem:approx2} can be written as
    \begin{equation}\small\vspace{-5pt}
         J_2 =\sum_{t=0}^{T} \min\left\{M_t^0,M_t^1,\dots,M_t^\delta\right\}.
    \end{equation}
   For each $t$, we set $\min\{M_t^0,\dots,M_t^\delta\}=\underline{M}_t$, and we have $M_t^d\geq \underline{M}_t,\forall d\in\{0,\dots,\delta\}$. 
   Plugging the inequalities into Equation~\eqref{eqn:j1}, we obtain $J_1\geq \min\{\sum_{t=0}^{T}\underline{M}_t,\dots,\sum_{t=0}^{T}\underline{M}_t\}=\sum_{t=0}^{T}\underline{M}_t=J_2$. 
   The equality only holds when $M^0_t=\dots=M^\delta_t$.  %So when we perform optimization on Problem \ref{problem:approx2}, we are maximizing the lower bound of optimal function value of $J_1$. 
    Furthermore, Equation~\eqref{eqn:bound} establishes the relation between $J_0$ and $J_1$, since $J_1$ is a direct LSE decomposition of $J_2$. Namely, we have $\frac{1}{\beta}\log J_1\leq J_0 +\frac{1}{\beta}\log(T+1)$. Thus $\frac{1}{\beta}\log J_2\leq J_0 +\frac{1}{\beta}\log(T+1)$ also holds as $J_2\leq J_1$, which means that $\frac{1}{\beta}\log J_2 - \frac{1}{\beta}\log(T+1)$ is a lower bound for the objective function value of the original Problem~\ref{problem:worst}. Thus maximizing $J_2$ maximizes the lower bound of $J_0$. The proof for $\Phi=\G_{[0,H)}\phi$ is analogous.
\end{proof}

\begin{figure*}[t!]
\centering
\vspace{-5pt}
\subfigure[]{\includegraphics[width=0.2\linewidth]{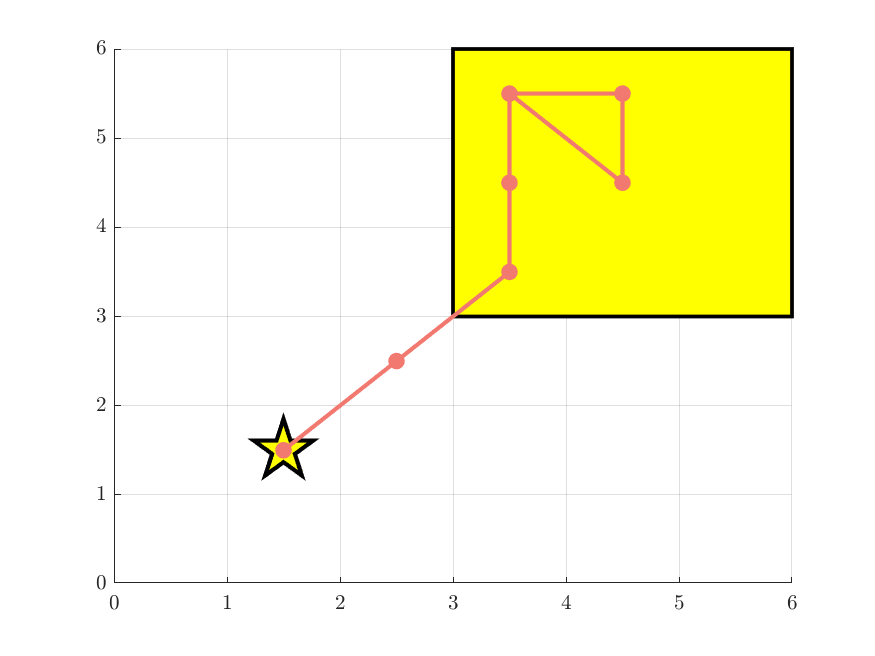}}
\subfigure[]{\includegraphics[width=0.2\linewidth]{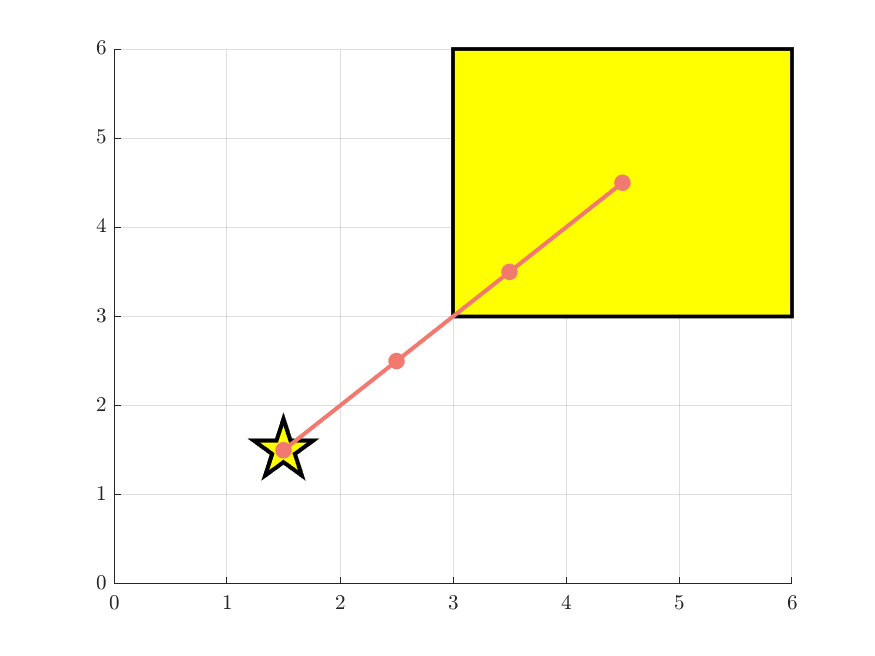}} 
\subfigure[]{\includegraphics[width=0.2\linewidth]{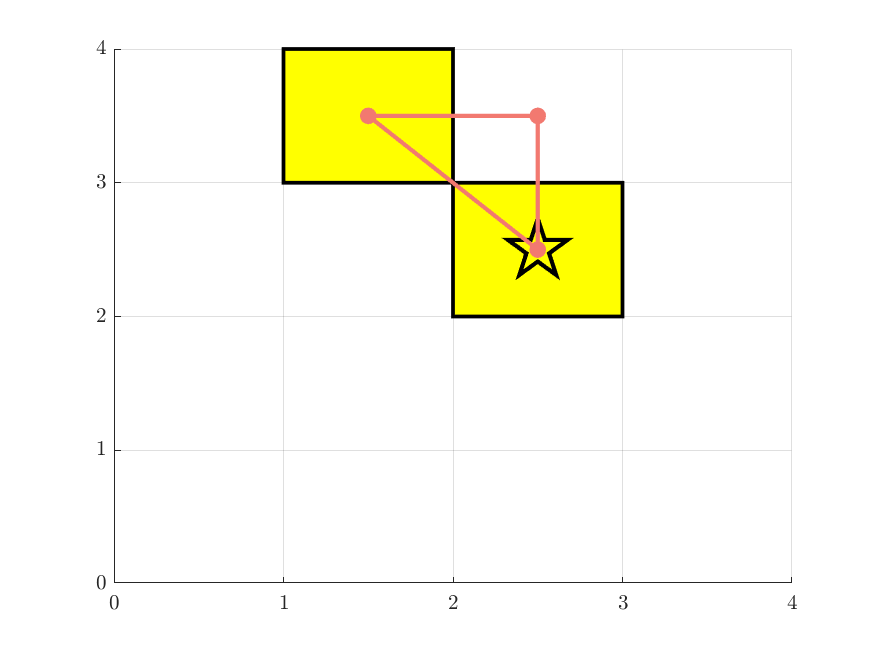}}
\subfigure[]{\includegraphics[width=0.2\linewidth]{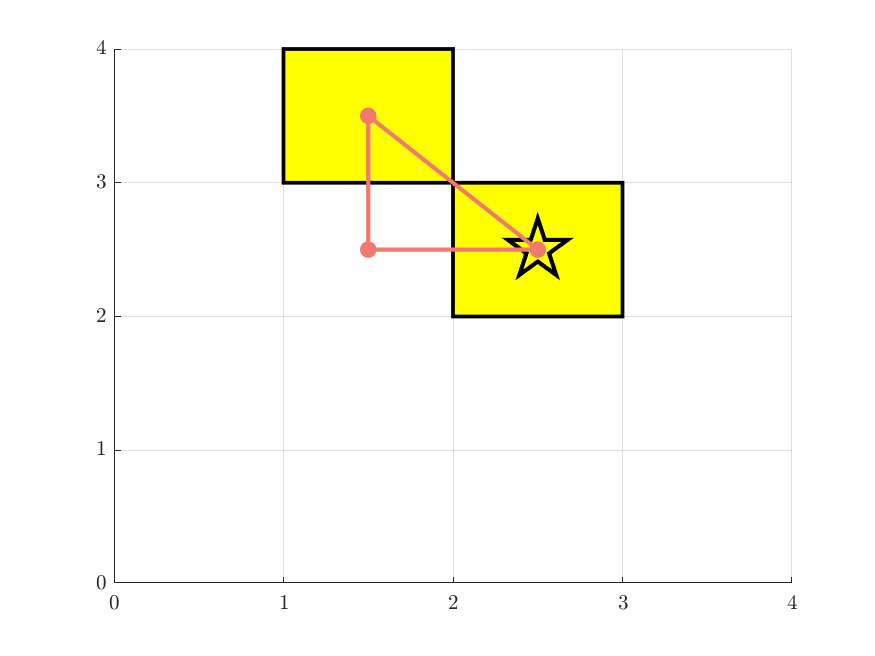}} \\
\vspace{-10pt}
\subfigure[]{\includegraphics[width=0.2\linewidth,trim=0 110 0 0,clip]{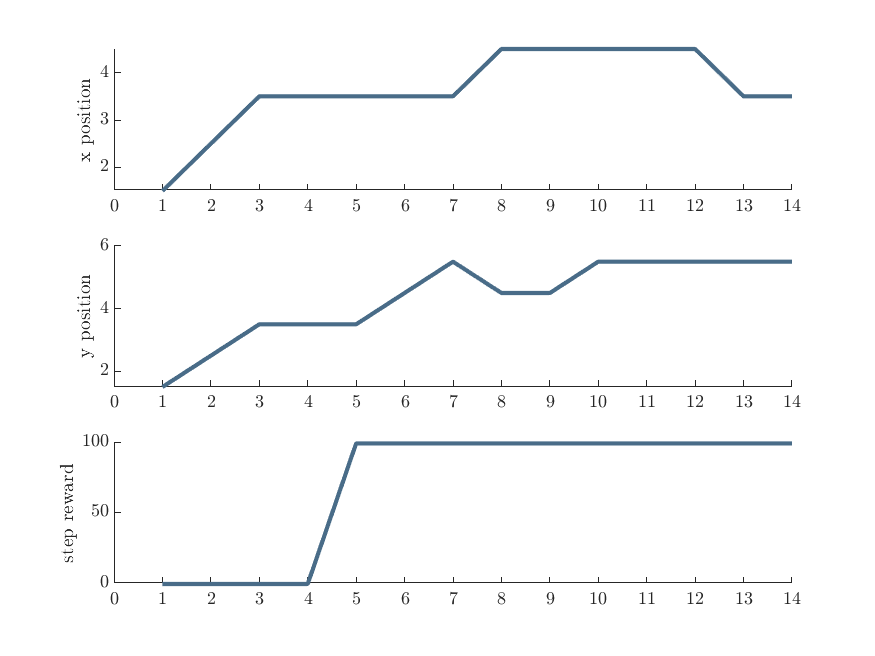}}
\subfigure[]{\includegraphics[width=0.2\linewidth,trim=0 110 0 0,clip]{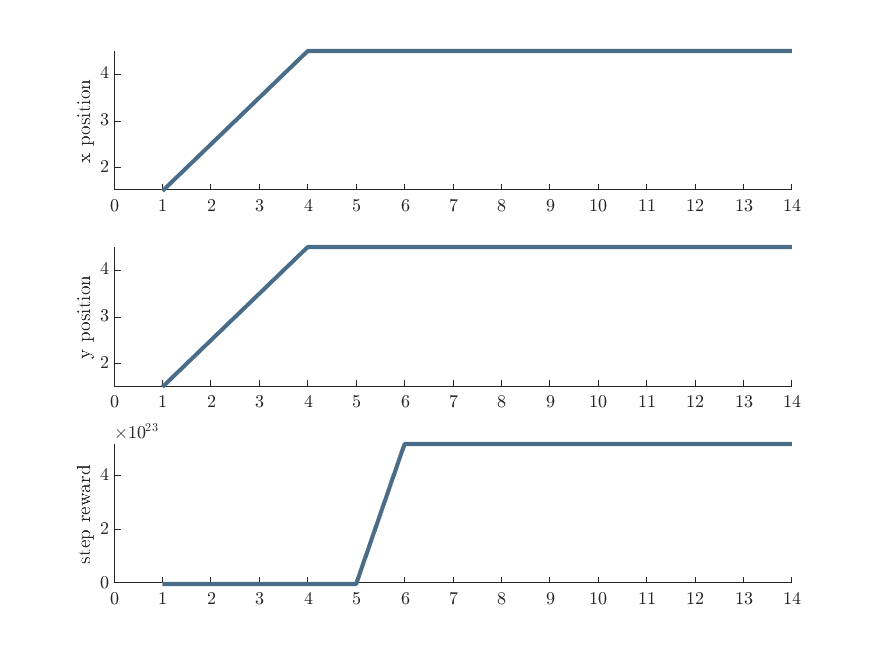}} 
\subfigure[]{\includegraphics[width=0.2\linewidth,trim=0 110 0 0,clip]{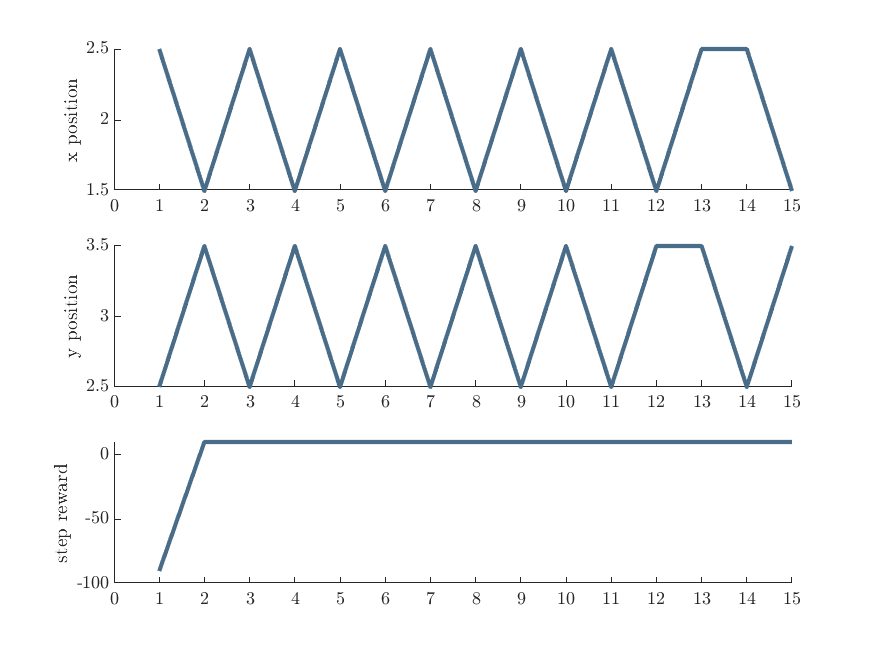}}
\subfigure[]{\includegraphics[width=0.2\linewidth,trim=0 110 0 0,clip]{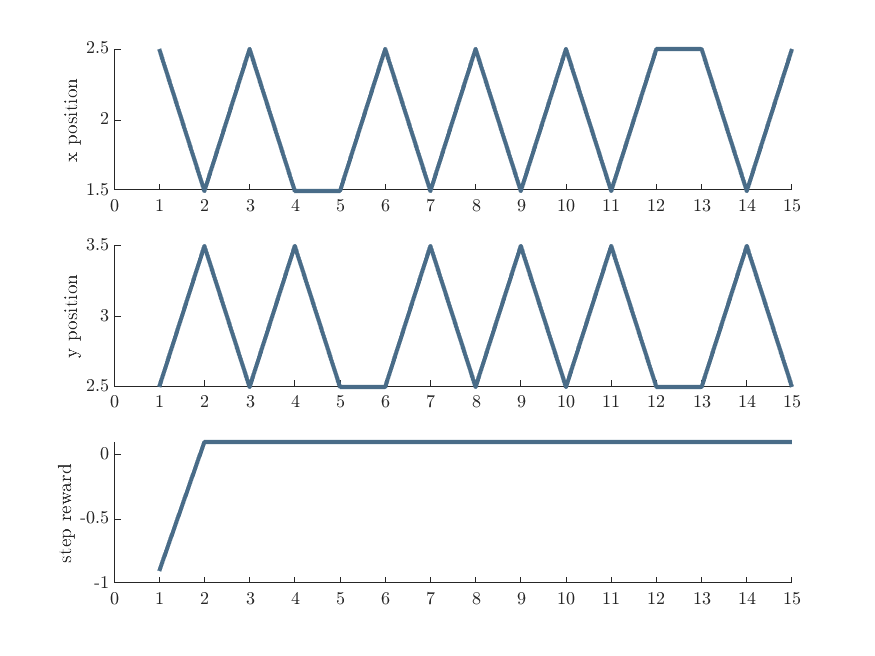}} \\
\vspace{-5pt}
\caption{(a)-(d) From left to right: Sample Trajectories generated by $\pi^*_{1A}$ for reachability task, $\pi^*_{2A}$ for reachability task, $\pi^*_{1A}$ for patrolling task
and $\pi^*_{2A}$ for patrolling task
\vspace{-12pt}
% (e)-(f) Projection on X and Y axis for (a)-(d)
}
% \SQ{I clipped the episodic reward sub-figure. Because I think the y-axis can cause confusion.}
\label{fig:exp}
\end{figure*}

\section{Case Studies}\label{sec:Sim}

In this section, we illustrate the effectiveness of our algorithm by conducting a set of four simulation experiments\footnote{Sample videos and codes are available at \small{\url{https://github.com/WSQsGithub/TimeRobustLearning}}.}. 
All algorithms were implemented using \texttt{MATLAB} and simulated in \texttt{Coppeliasim} on a Windows 10 computer with an Intel Core i7-8550U 1.80GHz processor.

\textbf{System Descriptions: }
We consider a scenario, where a warehouse robot navigates in a manufacturing factory floor as depicted in Figure~\ref{fig:scene}. 
There are three functional areas in the workspace:  a storage area marked by the parcel icon and two loading stations marked by $A$ and $B$.
These areas will be used  later for specifying  tasks.
Depending on the specific task, the workspace is abstracted as $n\times n$ grid worlds shown in Figure~\ref{fig:exp}, where star marks the initial state and yellow regions mark the goal states. 
We consider the high-level decision-making problem on the abstracted grid world, where agent can choose to move to its adjacent grids or to stay for each decision time step.  If the agent chooses an action leading to collisions with the boundaries, then it will be forced to stay at the same grid. 
The high-level decisions from grid  to grid are then executed by a low-level hybrid controller that ensures collision avoidance along the way.  
The objective is to synthesize a control policy that  generates trajectories over grids   satisfying a given STL task.
\vspace{-10pt}
 \begin{figure} [H]
    \centering
    \includegraphics[width=0.40\linewidth]{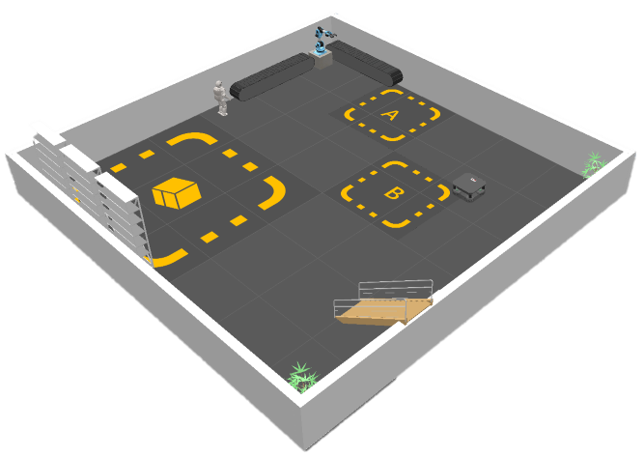}
    \caption{\small Aerial view of the factory floor}
    \label{fig:scene}
\end{figure}
\vspace{-12pt}

\textbf{Policy Learning Setups: }
We modify the standard tabular Q-learning algorithm to synthesize policies. To approximate reward functions, we choose $\beta=50$. The learning rate for the update is chosen as $\alpha=\max\{0.95\times0.999^i,0.0001\}$ at the $i$th episode and the discount factor is chosen as $\gamma=0.9999$. 
% To facilitate learning, rewards are scaled and subtracted a constant baseline to provide both positive and negative rewards during implementation. 
% Note that this scaling will not change the optimization problem since the result of the $\max$ operator is preserved. 
We settle for the policy after $10^4$ episodes of training. 

\textbf{Result Evaluations: }
Once a control policy is obtained, we evaluate it by generating 1000 simulation trajectories. 
Note that, we do not explicitly introduce time delays in the simulation, but evaluate each trajectory by its robustness to potential delays. 
For each trajectory, we compute the following four metrics by existing computation methods for STL: 
(i) Boolean satisfaction, (ii) {satisfaction with temporal robustness guarantees}, (iii) {spatial robustness value} and (iv) {temporal robustness value}. 
Then we compute the statistic values, including the average satisfaction rate $\text{Pr}(\s_{0:T}\models \Phi)$, the average time-robust satisfaction rate $\text{Pr}[\theta(\Phi,\s_{0:T})\leq \delta]$, the average spatial robustness $\bar{\rho}(\Phi,\s_{0:T})$ and the average temporal robustness $\bar{\theta}(\Phi,\s_{0:T})$ and summarize it in Table~\ref{tab:detail}. 
For each scenario, we pick one of the sample trajectories with the highest episodic reward for the purpose of demonstration as shown in Figure~\ref{fig:exp}. 

% \SQ{Replace the curve with a new column in the table?}
\begin{table}[H]
    \centering
    \caption{Experiment parameters and results}
    \adjustbox{max width=\linewidth}{
         \begin{threeparttable}
        \begin{tabular}{ccccccccc}
        \toprule
        Prob. & Task & Grid size & \#Q-entry  & $\delta$ & $\text{Pr}(\s_{0:T}\models \Phi)$&$\text{Pr}[\theta(\Phi,\s_{0:T})\geq \delta]$ &$\bar{\rho}(\Phi,\s_{0:T})$&$\bar{\theta}(\Phi,\s_{0:T})$\\
        \midrule
        1A & $\Phi_1$ & 6  & 12358 & 2 & 0.967 & 0.963 & 0.994 & 8.782\\
        2A & $\Phi_1$ & 6  & 13135 & 2 & 0.957 & 0.95 & 1.274 & 8.154\\
        1A & $\Phi_2$ & 4  & 4574 & 1 & 0.648 & 0.55 & 0.043 & -3.76\\
        2A & $\Phi_2$ & 4  & 4534 & 1 & 0.657 & 0.568 & 0.079 & -3.626\\
        \bottomrule
    \end{tabular}
    \end{threeparttable}
    }
    \label{tab:detail}
\end{table}\vspace{-20pt}
\begin{figure}[H] 
    \centering
    \includegraphics[width=0.4\linewidth,trim=20 20 20 30,clip]{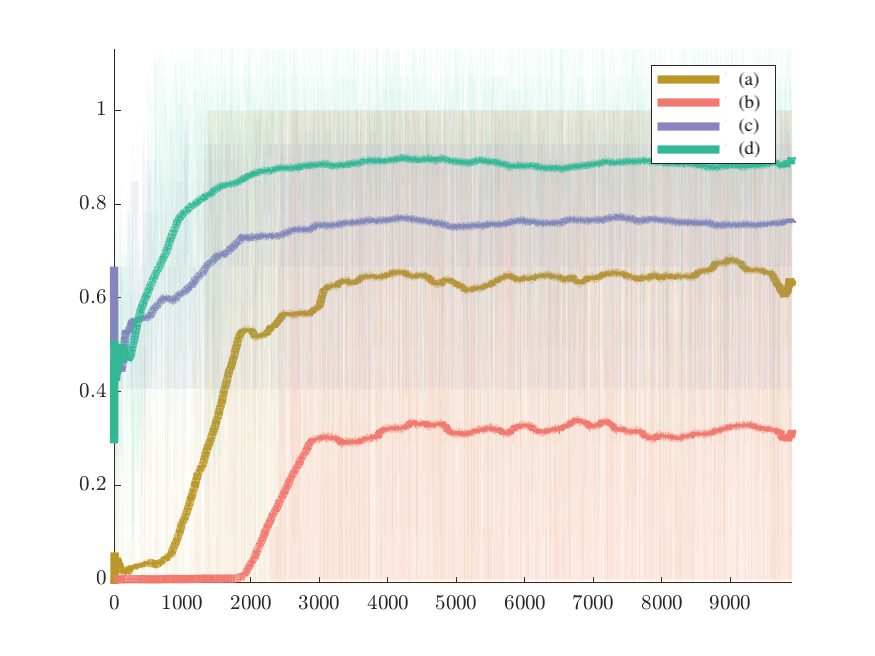}
    \caption{Learning curves of training the policies in Figure~\ref{fig:exp}(a-d), respectively}
    \label{fig:curve}
\end{figure}\vspace{-10pt} 

\textbf{Reachability Task: }
In the first case study, we consider a $6\times 6$ world and a reachability task described by\vspace{-2pt}
\[
\vspace{-2pt}\Phi_1 = \F_{[0,12)}\G_{[0,2)}(s\in Goal).
\] 
Specifically,  within $14$ steps, the agent needs to move to the storage area and stay for at least $2$ steps. 
We synthesize policies for both Problems~1A and 2A and show the sample trajectories in abstract space in Figures~\ref{fig:exp}(a) and (e) for Problem~1A and Figures~\ref{fig:exp}(b) and (f) for Problem~2A, respectively. 
The robot manages to find the shortest path to the goal region and stays there until the task ends even though the STL specification only requires it to stay there for $2$ steps.  Figure~\ref{fig:exp}(b) further shows that the robot is driven to the center of the goal region as a result of considering spatial robustness. In this case, the policy of Problem~\ref{problem:approx2} generates more spatially robust trajectories.

\textbf{Patrolling  Task: }
In the second case study, we consider a $4\times 4$ world and a patrolling task described by 
\[\vspace{-3pt}
\Phi_2 = \G_{[0,12)}\left[\F_{[0,3)}(s\in A)\land \F_{[0,3)}(s\in B)\right].
\]
Specifically, within $15$ time steps, the robot needs to visit regions $A$ and $B$ (two yellow grids) every $3$ time steps to load and unload workpieces. We still synthesize policies for both Problems~1A and 2A. 
Sample trajectories are provided in Figures~\ref{fig:exp}(c) and (g) for Problem~1A, and Figures~\ref{fig:exp}(d) and(h) for Problem~1B, respectively. 
For this task, since we take temporal robustness into account, the robot leaves the goal region immediately once the work piece is (un)loaded.
Compared with the reachability tasks, this patrolling task is more difficult to achieve since we need to satisfy the sub-task for the entire horizon, which explains the relatively low average temporal robustness compared to the reachability task, as indicated in Table~\ref{tab:detail}.

\textbf{Discussions: }
Table~\ref{tab:detail} shows that the performance distinction between policies from the two problem formulations is statistically insignificant. This is largely because that even though two problems are theoretically formulated differently, the shaped reward is numerically similar. Also, because $\beta$ has to be sufficiently large to approximate $\min$ or $\max$ operators, $e^{\beta x}$ becomes very large when $x>0$ and  close to $0$ otherwise. This scaling further narrows the gap between the two problems. Nevertheless, Figure~\ref{fig:exp} shows that the spatial robustness maximization drives the agent towards the goal region's center.

Figure~\ref{fig:curve} reveals intriguing insights from the learning curves. For the patrolling task, the potential of expediting learning through spatial robustness consideration is evident.  While for the reachability task, from Figure~\ref{fig:curve}, it seems that the optimum is later reached. 
However, this is because that the agent gets significantly higher reward at the central grid of the goal region than the others. The agent trained on Problem~\ref{problem:approx2} can already generate as good trajectories as the agent trained on Problem~\ref{problem:approx} before the latter policy reaches convergence. The fact that the central grid which generates the best reward when visited, is distant from the initial state, contributes to the late convergence as it takes longer to find the central grid.

\vspace{-5pt}
\section{Conclusion and Future Work} \label{sec:Con}
\vspace{-3pt}
In this work, we propose a novel reinforcement learning approach to enhance the temporal robustness of signal temporal logic tasks for unknown stochastic systems. 
We present two optimization problems to maximize temporal robustness probability and expected spatial-temporal robustness. 
Additionally, we provided approximation techniques that enable the application of standard Q-learning techniques.
Experimental results demonstrate the effectiveness of our proposed approach.
In the future, we plan to extend our results to more general fragments of STL tasks. 
We also aim to investigate how to incorporate the concept of asynchronous temporal robustness in reinforcement learning and to consider the case of continuous state and action space.

\clearpage
\bibliography{icra}
\bibliographystyle{ieeetr}
\end{document}